\documentclass[american,aps,pra,reprint, superscriptaddress]{revtex4-1}
\usepackage[unicode=true,pdfusetitle, bookmarks=true,bookmarksnumbered=false,bookmarksopen=false, breaklinks=false,pdfborder={0 0 0},backref=false,colorlinks=false] {hyperref}
\hypersetup{ colorlinks,linkcolor=myurlcolor,citecolor=myurlcolor,urlcolor=myurlcolor}
\usepackage{graphics,epstopdf,graphicx, amsthm, amsmath, amssymb, times, braket, colortbl, color, bm, framed, cleveref, mathrsfs}
\usepackage[up]{subfigure}
\definecolor{myurlcolor}{rgb}{0,0,0.7}

\newcommand{\tinyspace}{\mspace{1mu}}

\newcommand{\proj}[1]{| #1\rangle\!\langle #1 |}

\newcommand{\Br}[1]{\left[#1\right]}
\newcommand{\abs}[1]{\left\lvert\tinyspace #1 \tinyspace\right\rvert}
\newcommand{\norm}[1]{\left\lVert #1 \right\rVert}

\theoremstyle{plain}
\newtheorem{thm}{Theorem}
\newtheorem{lem}[thm]{Lemma}
\newtheorem{prop}[thm]{Proposition}
\newtheorem{cor}[thm]{Corollary}

\newcommand*{\myproofname}{Proof}

\def\ot{\otimes}
\def\complex{\mathbb{C}}

\def\cI{\mathcal{I}}
\def\cC{\mathcal{C}}

\def\cM{\mathcal{M}}

\def\cD{\mathcal{D}}
\def\sC{\mathscr{C}}
\def\sD{\mathscr{D}}

\def\rS{\mathrm{S}}
\def\cH{\mathcal{H}}

\makeatother

\begin{document}

 \author{Kaifeng Bu}
 \email{bkf@zju.edn.cn}
 \affiliation{School of Mathematical Sciences, Zhejiang University, Hangzhou 310027, PR~China}
  \author{Chunhe Xiong}
 \email{xiongchunhe@zju.edu.cn}
 \affiliation{School of Mathematical Sciences, Zhejiang University, Hangzhou 310027, PR~China}
%
 %
%

\title{A note on cohering power and de-cohering power }

\begin{abstract}
Cohering power and de-cohering power have recently been proposed to quantify the ability of
a quantum operation to produce and erase coherence respectively.
In this paper, we investigate the properties of cohering power and de-cohering power.
First, we prove the equivalence between two different kinds of cohering power   for any quantum operation on single qubit systems, which
implies that $l_1$ norm of coherence is monotone under Maximally incoherent operation (MIO) and Dephasing-covariant operation (DIO) in 2-dimensional space.
In higher dimensions, however, we show that the monotonicity under MIO or DIO  does not hold.
Besides, we compare the set of quantum operations with zero cohering power with Maximally incoherent operation (MIO) and Incoherent operation (IO).
Moreover, two different types of de-cohering power are defined and we find that they are not equal  in single qubit systems.
Finally, we make a comparison between cohering power and de-cohering power for single qubit unitary operations and show that
cohering power is always larger than de-cohering power.
\end{abstract}

\maketitle

\section{Introduction}
Quantum resource theory \cite{Oppenheim13, FBrandao15} plays an important role in  the development and quantitative understanding of various physical phenomena in quantum physics and quantum information theory. A resource theory consists of two basic elements: free operations and free states. Any operation (or state) is  dubbed as a resource if it falls out of the set of free operations (or the set of free states).  The most significant resource theory is the resource theory of quantum entanglement defined on bipartite or multipartite systems \cite{HorodeckiRMP09}, which is a basic resource for various information processing protocols including
superdense coding \cite{Bennett1992} and teleportation \cite{Bennett1993}.
However, for single quantum systems, quantum coherence, which is based
on the superposition rule, must be thought of a peculiar feature of quantum
mechanic just like entanglement in  bipartite systems.  Recently
significant advancements in fields like thermodynamic theory \cite{Rudolph214,Lostaglio2015,Brandao2015,Narasimhachar2015}, quantum biology \cite{Plenio2008,Lloyd2011,Levi14},
has suggested coherence to be a useful resource at the nanoscale, which leads to the development of the resource
theory of coherence
\cite{Baumgratz2014, Yadin2015, Aberg14, Bai2015, Bromley2015, Cheng2015, Chitambar2016,ChitambarPRA2016, ChitambarA2015, Chitambar2015, Du2015, Girolami14, Hu2016, Mani2015, Marvian14, MondalA2015, Napoli2016, Peng2015, Piani2016, Swapan2016, Fan15, Uttam2015, UttamA2015, UttamS2015, Singh2016, Zhang2015, Buc2015,BuE2016, AlexB2015, StreltsovA2015, Xi2015, Yao2015, Winter2016}.

One advantage of
having a resource theory for some physical quantity is the operational quantification of the relevant resources and the resource production through
a quantum operation.
In the resource theory of entanglement,   entangling power \cite{Zanardi2000} of
quantum operations has been proposed to  quantify the ability of quantum
operations to produce entanglement.
Besides,  cohering power and de-cohering power of  quantum operations have also been
proposed to quantify
the ability to produce coherence and erase coherence respectively \cite{Mani2015}. And it has been shown that
the cohering power of single  qubit unitary operations is equal to de-cohering power
in the skew information of coherence \cite{Girolami14}. Two different types of cohering power  have been defined
on the set of  incoherent states and the set of all quantum states respectively,  and it has been proved that
these two  types of cohering power are equal for unitary operation in single qubit case \cite{Bu2015,Diaz2015}. However,
whether this statement can be generalized to any quantum operation in single qubit case remains unclear.
In the present work, we further investigate cohering
power and de-cohering power.
And we prove  that these two types of cohering power  are equal for any quantum operation in 2-dimensional space, which
extends the result on unitary operations \cite{Bu2015,Diaz2015} to general quantum operations.
Besides,  as the cohering power of incoherent operations is always zero, we compare the
sets of quantum operations with zero cohering power with several different free operations for coherence \cite{Hu2016},  namely,  Incoherent operation (IO), Maximally incoherent operation (MIO) and Dephasing-covariant incoherent operation (DIO) \cite{Baumgratz2014,Chitambar2016,ChitambarPRA2016}.
As free operations cannot increase the amount of the relevant resource,  the monotonicity of resource measure under free operations
is crucial to the resource theory. Whether $l_1$ norm of coherence is monotone under MIO and DIO or not is   an open problem
proposed in \cite{Chitambar2016,ChitambarPRA2016}.
In this  work,  we prove that   $l_1$ norm
of coherence is not monotone under MIO or DIO.
Due to this statement, we  demonstrate  the operational gap between DIO and IO in terms of state transformation, which is also an  open problem proposed in
\cite{Chitambar2016,ChitambarPRA2016}.
Furthermore, we derive the exact expression
 for de-cohering power of unitary operations on single qubit systems. Two different kinds of de-cohering power have also  been
 defined
on the set of  maximally coherent  states and the set of all quantum states respectively. We also compare these two  kinds of de-cohering power but find they are not equal in single qubit systems, which is
different from the cohering power.
Finally, we make a comparison between the cohering power and de-cohering power and find that
de-cohering power is always less than the cohering power for unitary operations on single qubit systems.

This work is organized as follows. In Sec.\ref{sec:pre}, we provide
the preliminary material in the resource theory of coherence. We investigate two types of cohering power are equal for any quantum operation in single qubit case.
And  we show that there is no monotonicity for $l_1$ norm of  coherence under MIO  or DIO in
Sec.\ref{sec:coh}. Besides, we derive the explicit formula for de-cohering power and compare two different types of  de-cohering power in Sec.\ref{sec:dech}.
Moveover, we compare the cohering power and the de-cohering power in 2-dimensional space in
Sec.\ref{sec:comp}. Finally, we conclude in Sec.\ref{sec:con}.

\section{Preliminary and notations}\label{sec:pre} 
\noindent
{\it Free states and free operations in the resource theory of coherence ( see \cite{Baumgratz2014} and \cite{Chitambar2016,ChitambarPRA2016})--}
Given a fixed reference basis, say $\{\ket{i}\}$, any state which is diagonal in the reference basis is called an incoherent state. And the   set of all incoherent states is denoted by $\mathcal{I}$. Then we introduce several
different free operations in the resource theory of coherence  from \cite{Baumgratz2014,Chitambar2016,ChitambarPRA2016}.

(1) Incoherent operation (IO).
A quantum operation $\Phi$ is called an incoherent operation if there exists a set of Kraus operators $\{K_n\}$ of $\Phi$  such that $K_n\mathcal{I} K^{\dag}_n\subset \mathcal{I}$ for any $n$.

(2) Maximally incoherent operation (MIO).
A quantum operation $\Phi$ is called a maximally incoherent operation if
$\Phi(\cI)\subset \cI$.

(3) Dephasing-covariant incoherent operation (DIO).
A quantum operation $\Phi$ is called a Dephasing-covariant incoherent operation if
\begin{eqnarray}
\Br{\Delta,\Phi}=0,
\end{eqnarray}
where $\Delta(\rho):=\sum_{i}
\bra{i}\rho \ket{i}\ket{i}\bra{i}$.

\noindent
{\it $l_1$ norm and relative entropy measure (see \cite{Baumgratz2014})-- }
\begin{enumerate}
\item[(i)] $ l_{1}$ norm measure $\cC_{l_{1}}$ is defined by
\begin{eqnarray}
\cC_{l_{1}}(\rho):=\sum_{i\neq j}|\rho_{ij}|.
\end{eqnarray}
\item[(ii)] Relative entropy measure $\cC_{r}$ is defined by
\begin{eqnarray}
\cC_{r}(\rho):= \rS(\rho^{(d)}) - \rS(\rho),
\end{eqnarray}
where $\rS(\rho)=-\mathrm{Tr}{\rho\log\rho}$ is the von Neumann entropy of
$\rho$ and $\rho^{(d)}$ is the diagonal state of $\rho$.
\end{enumerate}

\noindent
{\it Cohering power--}
Two types of cohering power (see \cite{Mani2015} and \cite{Bu2015}):
\begin{eqnarray}
\sC_X(\Phi):&=&\max_{\rho\in\cI}\set{\cC_X(\Phi(\rho))},\\
\widehat{\sC}_X(\Phi):&=&\max_{\rho \in \cD(\cH)}\set{\cC_X(\Phi(\rho))-\cC_X(\rho)}
\end{eqnarray}
where $X$ denotes a  coherence measure  and $\cI$ is the set of
incoherent states. To distinguish these two powers, we call $\sC$ and
$\widehat{\sC}$ the cohering power and generalized cohering power, respectively.
Obviously,  $\sC_X(\Phi)\leq \widehat{\sC}_X(\Phi)$ for any coherence measure X.

\noindent
{\it Formula of cohering power for unitary operations (see \cite{Bu2015} )--}
It has been shown in \cite{Bu2015} that the cohering power for a unitary operation $U=[U_{ij}]_{d\times d}$ can be written as
\begin{eqnarray}
\sC_{l_1}(U)=\norm{U}^2_{1\to 1}-1,
\end{eqnarray}
where $\norm{U}_{1\to 1}=\max\Set{\sum^d_{i=1} \left|U_{ij}\right|:j=1,\ldots,d}$.
And
\begin{eqnarray}
\sC_{r}(U)=\max\set{S(|U_{1i}|^2, |U_{2i}|^2,
\cdots, |U_{di}|^2),i\in[d]},~~~~~~~~
\end{eqnarray}
where $S(\set{p_i})=\sum -p_i\log p_i$.

\noindent
{\it  De-cohering power (see \cite{Mani2015})--}
Two types of decohering power:
\begin{eqnarray}
\sD_X(\Phi):&=&\max_{\rho\in\cM}\set{\cC_X(\rho)-\cC_X(\Phi(\rho))},\\
\widehat{\sD}_X(\Phi):&=&\max_{\rho \in \cD(\cH)}\set{\cC_X(\rho)-\cC_X(\Phi(\rho))}.
\end{eqnarray}
where $X$ denotes a  coherence measure  and $\cM$ is the set of
maximally coherent states. To distinguish them, we call $\sD$ and
$\widehat{\sD}$ the de-cohering power and generalized de-cohering power, respectively.
Clearly,  $\sD_X(\Phi)\leq \widehat{\sD}_X(U)$ for any coherence measure X.
Note that maximally coherent state must be pure state and can be expressed as 
$\ket{\psi}=\frac{1}{\sqrt{d}}\sum_ke^{i\theta_k}\ket{k}$ \cite{Peng2015}.

\section{Results about cohering power}\label{sec:coh}
In view of the definitions,  cohering power and generalized cohering power are different in essence: one is defined on the set
of incoherent states and the other is defined on the set of all quantum states.
As can be seen,  cohering power is always less than
the generalized cohering power. Moreover,
it has been proved that for any unitary operation $U$ on a single qubit system, the cohering
power and  the generalized cohering power coincides, that is,
$\sC_{l_1}(U)=\widehat{\sC}_{l_1}(U)$ \cite{Bu2015}. This means the
maximal coherence produced by unitary operation over all states can be obtained
by considering only the incoherent states which is a smaller set of states.
Here, we generalize this statement to  any quantum operation  $\Phi$ on single
qubit systems.

\begin{prop}\label{prop:e_l1}
For any quantum operation $\Phi$ on a single qubit system, the cohering
power and  the generalized cohering power coincides, that is,
$\sC_{l_1}(\Phi)=\widehat{\sC}_{l_1}(\Phi)$.
\end{prop}
\begin{proof}
For any quantum operation $\Phi$ on a single qubit system,  it can be
expressed by a set of Kraus operators $\set{K_n}_n$ as
\begin{eqnarray*}
\Phi(\cdot)=\sum_n K_n\cdot K^\dag_n,
\end{eqnarray*}
where $ K_n=\Br{\begin{array}{ccc}
K^{(1,1)}_n & K^{(1,2)}_n \\
K^{(2,1)}_n & K^{(2,2)}_n
\end{array}}
$ and $\sum_nK^\dag_nK_n=\mathbb{I}$. Any qubit state $\rho$ can be written as
$\rho=\frac{\mathbb{I}}{2}+\frac{1}{2} {\vec{\sigma}}\cdot \vec{r}$, where $\vec{r}=(x, y, z)$ is a unit vector and
${\vec{\sigma}}=(\sigma_x, \sigma_y, \sigma_z)$ is the Pauli
matrices.
Thus, the $l_1$ norm of coherence of initial state $\rho$ and final state
$\Phi(\rho)$ are specified by
\begin{eqnarray*}
\cC_{l_{1}}(\rho)=\abs{x+iy},
\end{eqnarray*}
and
\begin{eqnarray*}
\cC_{l_{1}}(\Phi(\rho))&=&\left|\sum_n[\overline{K^{(2,1)}_n}K^{(1,1)}_n(1+z)+
\overline{K^{(2,2)}_n}K^{(1,2)}_n(1-z)\right.\nonumber\\
&&\left.+\overline{K^{(2,1)}_n}K^{(1,2)}_n(x-i y)
+\overline{K^{(2,2)}_n}K^{(1,1)}_n(x+i y)]
\right|.~~~~~~~
\end{eqnarray*}
Since the cohering power is only defined on incoherent states,  then cohering power of $\Phi$ can be written as
\begin{eqnarray*}
\sC_{l_1}(\Phi)=2\max\left\{\abs{\sum_n\overline{K^{(2,1)}_n}K^{(1,1)}_n},
\abs{\sum_n\overline{K^{(2,2)}_n}K^{(1,2)}_n}\right\}.
\end{eqnarray*}
Since
\begin{eqnarray*}
&&\cC_{l_{1}}(\Phi(\rho))\\
&\leq& \abs{\sum_n\overline{K^{(2,1)}_n}K^{(1,1)}_n(1+z)}+\abs{\sum_n\overline{K^{(2,2)}_n}K^{(1,2)}_n(1+z)}\\
&+&\abs{\sum_n\overline{K^{(2,1)}_n}K^{(1,2)}_n}\abs{x-iy}+\abs{\sum_n\overline{K^{(2,2)}_n}K^{(1,1)}_n}\abs{x+iy}\\
&\leq& 2\max\left\{\abs{\sum_n\overline{K^{(2,1)}_n}K^{(1,1)}_n},
\abs{\sum_n\overline{K^{(2,2)}_n}K^{(1,2)}_n}\right\}\\
&+&\abs{\sum_n\overline{K^{(2,1)}_n}K^{(1,2)}_n}\abs{x-iy}+\abs{\sum_n\overline{K^{(2,2)}_n}K^{(1,1)}_n}\abs{x+iy}\\
&\leq& 2\max\left\{\abs{\sum_n\overline{K^{(2,1)}_n}K^{(1,1)}_n},
\abs{\sum_n\overline{K^{(2,2)}_n}K^{(1,2)}_n}\right\}\\
&+&\sum_n \frac{\sum^2_{i,j=1}\abs{K^{(i,j)}_n}^2}{2} \abs{x+i y}\\
&=&2\max\left\{\abs{\sum_n\overline{K^{(2,1)}_n}K^{(1,1)}_n},
\abs{\sum_n\overline{K^{(2,2)}_n}K^{(1,2)}_n}\right\}\\
&+&|x+iy|,
\end{eqnarray*}
then
\begin{eqnarray*}
&&\cC_{l_{1}}(\Phi(\rho))-\cC_{l_1}(\rho)\\
&\leq & 2\max\left\{\abs{\sum_n\overline{K^{(2,1)}_n}K^{(1,1)}_n},\abs{\sum_n\overline{K^{(2,2)}_n}K^{(1,2)}_n}\right\}\\
&\leq&\sC_{l_1}(\Phi),
\end{eqnarray*}
which implies
\begin{eqnarray*}
\widehat{\sC}_{l_1}(\Phi)\leq \sC_{l_1}(\Phi).
\end{eqnarray*}
Therefore, $\widehat{\sC}_{l_1}(\Phi)= \sC_{l_1}(\Phi)$
for any quantum operation $\Phi$ on qubit system.
\end{proof}

The above proposition is also an evidence that cohering power $\sC_{l_1}$ can be used to quantify the ability of a quantum operation
to generate coherence even if it is only defined on  incoherent states.
Besides, this result can  be used to demonstrate the monotonicity of $l_1$ norm of coherence
under  DIO and MIO in single qubit system directly. However,  monotonicity of $l_1$ norm coherence
under DIO and MIO does not hold in higher dimensional space.

\begin{prop}[Non-monotonicity for $l_1$ norm of coherence under DIO and MIO]\label{thm:no_mont}
In single qubit system, the $l_1$ norm of coherence can not increase under DIO and MIO. However, such
statement is not true in N-qubit system with $N\geq 2$, that is, there exists  a state $\rho_N\in D(\complex^{\otimes N})$  and a DIO (or MIO) $\Phi_N$ such
that $\cC_{l_1}(\Phi_N(\rho_N)) > \cC_{l_1}(\rho_N)$.
\end{prop}
\begin{proof}
Due to the definition of cohering power, it is easy to see that $\sC_{l_1}(\Phi)=0$ is equivalent to $\Phi(\cI)\subset \cI$, which means
that such $\Phi$ is a MIO. Due to Proposition \ref{prop:e_l1}, we have $\widehat{\sC}_{l_1}(\Phi)=0$ for any MIO $\Phi$ on a single qubit system. Thus, the $l_1$ norm of coherence can not increase under MIO. Since $DIO\subset MIO$, then we also
have the monotone of $l_1$ norm of coherence under DIO in single qubit case.

Next,  we show there exists a DIO $\Phi$ and a state a state $\rho$ such
that $\cC_{l_1}(\Phi(\rho)) > \cC_{l_1}(\rho)$ in 2-qubit system.
Consider the quantum operation $\Phi$ with following Kraus operators
\begin{equation*}
M_1=\Br{\begin{array}{cccc}
0 & \frac{1}{2} & 0 & 0  \\
\frac{1}{2\sqrt{3}} & 0 & 0& 0\\
-\frac{1}{2\sqrt{3}}& 0 & 0& 0\\
\frac{1}{2\sqrt{3}} & 0 & 0& 0
\end{array}},
M_2=\Br{\begin{array}{cccc}
\frac{1}{2\sqrt{3}} & 0 & \frac{1}{\sqrt{2}} &  \frac{1}{\sqrt{6}}  \\
0 & \frac{1}{2} & 0& 0\\
\frac{1}{2\sqrt{3}}& 0 & 0& 0\\
\frac{1}{2\sqrt{3}} & 0 & 0& 0
\end{array}},
\end{equation*}

\begin{equation*}
M_3=\Br{\begin{array}{cccc}
\frac{1}{2\sqrt{3}} & 0 & -\frac{1}{\sqrt{2}} &  \frac{1}{\sqrt{6}}  \\
\frac{1}{2\sqrt{3}} & 0 & 0& 0\\
0 & \frac{1}{2} & 0& 0\\
-\frac{1}{2\sqrt{3}} & 0 & 0& 0
\end{array}},
M_4=\Br{\begin{array}{cccc}
\frac{1}{2\sqrt{3}}& 0 & 0 & -\frac{\sqrt{6}}{3}  \\
-\frac{1}{2\sqrt{3}} & 0 & 0& 0\\
-\frac{1}{2\sqrt{3}}& 0 & 0& 0\\
0 & \frac{1}{2} & 0& 0
\end{array}},
\end{equation*}
It can be easily  verified
such operation $\Phi$ is a DIO according to  \cite{Chitambar2016,ChitambarPRA2016}.
Besides, let us take the state as following
\begin{equation*}
\rho=\Br{\begin{array}{cccc}
\rho_{11} & \rho_{12} & 0 & 0\\
\rho_{21} & \rho_{22} & 0 & 0\\
0 & 0 & 0 & 0\\
0 & 0 & 0 & 0
\end{array}}
\end{equation*}
with $\rho_{12}=\rho_{21}>0$. Then, through some calculation,
 $\cC_{l_1}(\Phi(\rho))=\frac{4}{\sqrt{3}}\rho_{12}$, which is
lager than $\cC_{l_1}(\rho)=2\rho_{12}$. Furthermore,
 for any N qubit system with $N\geq 3$, let us take
 $\Phi_N=\Phi\otimes \mathbb{I}_{N-2}$ and $\rho_{N}=\rho\otimes \sigma_{N-2}$
 where $\mathbb{I}_{N-2}$ denotes the identity operator on the remaining
 (N-2)-qubit system and $\sigma_{N-2}$ is a state of the remaining
 (N-2)-qubit system with $\cC_{l_1}(\sigma_{N-2})>0$.
 It is easily to see that such  $\Phi_N$ is also a
 DIO. Thus,
 \begin{eqnarray*}
 &&\cC_{l_1}(\Phi_N(\rho_N))- \cC_{l_1}(\rho_N)\\
 &=&\cC_{l_1}(\Phi(\rho))\ot \sigma_{N-2})
 -\cC_{l_1}(\rho\ot\sigma_{N-2})\\
 &=&[\cC_{l_1}(\Phi(\rho))-\cC_{l_1}(\rho)][\cC_{l_1}(\sigma_{N-2})+1]\\
 &>&\cC_{l_1}(\Phi(\rho))-\cC_{l_1}(\rho)>0,
 \end{eqnarray*}
 where the second equality comes from the
 multiplicity of $l_1$ norm of coherence, that is
 $\cC_{l_1}(\tau_1\ot\tau_2)+1=[\cC_{l_1}(\tau_1)+1][\cC_{l_1}(\tau_2)+1]$ for any two states $\tau_1$ and $\tau_2$.
 Thus, the $l_1$ norm of coherence is not monotonous under
 DIO in N-qubit system with $N\geq 2$. Since DIO is a subset of MIO,
 it also implies that there is no monotonicity of $l_1$ norm coherence under MIO.

\end{proof}

\begin{cor}
There exists state transformation $\rho\to\sigma$ by DIO which is not possible by IO.
\end{cor}
\begin{proof}
Let us take the states $\rho$ and $\Phi(\rho)$ given in the Proof of Proposition \ref{thm:no_mont},
then state transformation  $\rho\longrightarrow \sigma=\Phi(\rho)$ is feasible by DIO, but
not possible by IO, as $\cC_{l_1}(\Phi(\rho))>\cC_{l_1}(\rho)$ and
IO can not increase coherence of the states.
\end{proof}
This corollary shows the operational gap between
DIO and IO
in terms of state transformation
which is an open problem proposed in \cite{Chitambar2016,ChitambarPRA2016}.
Besides, the non-monotonicity
of $l_1$ norm coherence under MIO implies that $l_1$ norm is not contracting
under CPTP maps.  Contracting under CPTP maps is an important
property of norms as any norm with such property can usually be  used
as  a potential coherence quantifier in  the resource theory of coherence \cite{Baumgratz2014}. It is striking
 that $l_1$ norm can be employed to quantify coherence although it does not have
 such property.

\begin{cor}
$l_1$ norm is not contracting under CPTP maps, that is,
there exists quantum states $\rho$, $\sigma$ and CPTP map $\Phi$ such that $\norm{\Phi(\rho)-\Phi(\sigma)}_{l_1}>\norm{\rho-\sigma}_{l_1}$, where
$\norm{\rho}_{l_1}:=\sum_{i,j}|\rho_{ij}|$.
\end{cor}
\begin{proof}
If $l_1$ norm is contracting under CPTP maps, then for any quantum state $\rho$ and any MIO $\Phi$,
\begin{eqnarray*}
\cC_{l_1}(\rho)&=&\min_{\sigma\in\cI}\norm{\rho-\sigma}_{l_1}\\
&\geq &\min_{\sigma\in \cI}\norm{\Phi(\rho)-\Phi(\sigma)}_{l_1}\\
&\geq& \min_{\sigma\in \cI}\norm{\Phi(\rho)-\sigma}_{l_1}\\
&=&\cC_{l_1}(\Phi(\rho)),
\end{eqnarray*}
which contradicts with Proposition \ref{thm:no_mont}.
\end{proof}

In fact, as the cohering power $\cC_{l_1}$ and $\cC_r$ are both defined on the set of incoherent states $\cI$,
it is easy to see that the quantum operations with zero cohering power in $l_1$ norm of coherence or
 relative entropy of coherence is MIO, that is $ MIO=\set{\Phi:\sC_{l_1}(\Phi)=0}=\set{\Phi:\sC_r(\Phi)=0}$, which means that MIO is the set of all operation that can not increase the coherence of incoherent states.  We also consider the quantum operations with zero generalized cohering power as following,
\begin{eqnarray}
NIO_{l_1}&=&\set{\Phi:\widehat{\sC}_{l_1}(\Phi)=0},\\
NIO_r&=&\set{\Phi:\widehat{\sC}_r(\Phi)=0}.
\end{eqnarray}
Note  that the set  $NIO_{l_1}$( resp. $NIO_r$) is the set of all quantum  operations that will not increase the coherence of all states in
$l_1$ norm of coherence (resp. relative entropy of coherence).
Due to  the definition of generalized cohering power, we have
$NIO_r\subset MIO$.  Since relative entropy of coherence  is monotone under MIO \cite{Chitambar2016,ChitambarPRA2016}, then $MIO\subset NIO_r$, which implies that $MIO=NIO_r$. That is,
MIO is just the set of all quantum  operations that will not increase the coherence of all quantum states in relative entropy measure.  Moreover, we
get the relationship between IO, MIO, $NIO_{l_1}$ and $NIO_r$.

\begin{cor}
The relationship between IO, MIO, $NIO_{l_1}$ and $NIO_r$ in N-qubit system ( $N\geq 2$) is
\begin{eqnarray}
IO \subsetneq  NIO_{l_1}\subsetneq MIO=NIO_r
\end{eqnarray}
However, in single qubit system, the  relationship will become
\begin{eqnarray}
IO\subsetneq NIO_{l_1}=MIO=NIO_r.
\end{eqnarray}
\end{cor}
\begin{proof}

Since $\sC_{l_1}(\Phi)=\widehat{\sC}_{l_1}(\Phi)$ in single qubit system, then
$NIO_{l_1}=MIO$ due to the definition of $NIO_{l_1}$ and the fact $MIO=\set{\Phi:\sC_{l_1}(\Phi)=0}$.
Besides,  it has been demonstrated that there exists a quantum operation on single qubit system $\Phi\in MIO$ but $\Phi\notin IO$ (see \cite{Hu2016} and the Erratum of \cite{ChitambarPRA2016}). Thus $IO\subsetneq NIO_{l_1}=MIO=NIO_r$.

In  N-qubit system ( $N\geq 2$),  $NIO_{l_1}\subsetneq MIO$
comes from    Proposition \ref{thm:no_mont}. Thus, the relationship between IO, MIO and $NIO_{l_1}$ in N-qubit system ( $N\geq 2$) will become
$
IO \subsetneq  NIO_{l_1}\subsetneq MIO=NIO_r
$.

\end{proof}

The above proposition tells us that in single qubit system, MIO is also the set of quantum operations that will not 
increase the coherence of all quantum states in the $l_1$ norm measure, that is,  $NIO_{l_1}$ and $NIO_r$ coincides in this case.
The relationship between these sets may help us understand the role of IO and MIO in the resource theory
of coherence and be complementary to the previous work \cite{Chitambar2016,ChitambarPRA2016}. Besides,
since the relationship between $l_1$ norm of coherence and relative entropy coherence has been considered in \cite{Swapan2016}, we also consider the
relationship between cohering power defined in $l_1$ norm $\sC_{l_1}$ and that defined in relative entropy  $\sC_r$ for unitary operations.
\begin{prop}
Given a unitary operation $U$ in d-dimensional space, we have
\begin{eqnarray}\label{re:l1re}
\sC_{l_1}(U)\geq \max\set{\sC_r(U),2^{\sC_r(U)}-1}.
\end{eqnarray}
\end{prop}
\begin{proof}
Since  $l_1$ norm coherence and relative entropy coherence in pure states has the the following relationship  $\cC_{l_1}(\ket{\psi})\geq \max\set{\cC_r(\ket{\psi}),2^{\cC_r(\ket{\psi})}-1}$\cite{Swapan2016},
 it is easy to see the cohering power $\sC_{l_1}(U)=\max\set{\cC_{l_1}(U\ket{i}):i=1,...,d}$ and $\sC_{r}(U)=\max\set{\cC_r(U\ket{i}):i=1,...,d}$ also satisfy this relationship, that is,
\begin{eqnarray*}
\sC_{l_1}(U)\geq \max\set{\sC_r(U),2^{\sC_r(U)}-1}.
\end{eqnarray*}
\end{proof}
However, whether the cohering power of any quantum operation $\Phi$ satisfy \eqref{re:l1re} is still a question, which is closely related to the open problem:
the  potential relationship  between $l_1$ norm of coherence $\cC_{l_1}$ and relative entropy of coherence $\cC_r$  \cite{Swapan2016}.

\section{Results about de-cohering power}\label{sec:dech}
As mentioned before,  de-cohering power and generalized de-cohering power are defined by the maximization
over the set of maximally coherent states and all quantum states respectively. As both sets contain too many
states, it is difficult to calculate the exact value of de-cohering power and generalized de-cohering power of a given quantum operation.
 Here, we consider a simple case and  give the exact formula
of de-cohering power  and generalized de-cohering power for unitary operations in single qubit case, which makes the comparison
between  de-cohering power  and generalized de-cohering power possible.

\begin{prop}\label{fom:de}
For a qubit unitary operation $U$, which can be expressed as (up to a phase
factor)
$ U= \Br{\begin{array}{ccc}
a & b  \\
-b^{\ast} & a^{\ast}
\end{array}}
$ where $|a|^{2}+|b|^{2}=1$, the de-cohering power
 in $l_1$ norm of coherence and relative entropy of coherence
 can be expressed as
\begin{eqnarray}
\label{eq:dl1}\sD_{l_1}(U)
&=&1-||a|^2-|b|^2|\\
\label{eq:dr}\sD_r(U)&=&1-S(\frac{1}{2}+|ab|,\frac{1}{2}-|ab|)
\end{eqnarray}
And the generalized de-cohering power of $U$ is equal to
the generalized cohering power of $U^\dag$, that is
\begin{eqnarray}
\label{eq:cdl1}\widehat{\sD}_{l_1}(U)&=&\widehat{\sC}_{l_1}(U^\dag)\\
\label{eq:cdr}\widehat{\sD}_{r}(U)&=&\widehat{\sC}_{r}(U^\dag)
\end{eqnarray}

\end{prop}
\begin{proof}
In single qubit system, the maximal coherent state
can be written as $\ket{\psi}=(\frac{1}{\sqrt{2}},
\frac{1}{\sqrt{2}}e^{i\theta})^t$, where t denotes transposition. Then $U\ket{\psi}$
would be $\frac{1}{\sqrt{2}}(a+be^{i\theta}, -b^*+a^*e^{i\theta})^t$.
Thus
\begin{eqnarray*}
\sD_{l_1}(U)&=&1-\min_{\theta}|(a+be^{i\theta})(-b^*+a^*e^{i\theta})|\nonumber\\
&=&1-||a|^2-|b|^2|.
\end{eqnarray*}
Denote $a=|a|e^{i\theta}$ and $b=|b|e^{i\theta_b}$, then
\begin{eqnarray*}
\sD_{r}(U)&=&1-\min_{\vartheta}S(\frac{1}{2}+|ab|\cos\vartheta, \frac{1}{2}-|ab|\cos\vartheta)\nonumber\\
&=&1-S(\frac{1}{2}+|ab|,\frac{1}{2}-|ab|),
\end{eqnarray*}
where $\vartheta=\theta+\theta_b-\theta_a$.

Besides, in view of the definition of generalized de-cohering power
\begin{eqnarray*}
\widehat{\sD}_{l_1}(U)&=&\max_{\rho \in \cD(\cH)}\set{\cC_{l_1}(\rho)-\cC_{l_1}(U\rho U^\dag)}\\
&=&\max_{\rho \in \cD(\cH)}\set{\cC_{l_1}(U^\dag (U\rho U^\dag) U)-\cC_{l_1}(U\rho U^\dag)}\\
&=&\max_{U\rho U^\dag \in \cD(\cH)}\set{\cC_{l_1}(U^\dag (U\rho U^\dag) U)-\cC_{l_1}(U\rho U^\dag)}\\
&=&\widehat{\sC}_{l_1}(U^\dag).
\end{eqnarray*}

And $\widehat{\sD}_{r}(U)=\widehat{\sC}_{r}(U^\dag)$ can be obtained in a similar way.

\end{proof}

As can be seen from \eqref{eq:cdl1} and \eqref{eq:cdr},
 the amount of   coherence  produced by a unitary operation $U$
 is equal to that of coherence erased by $U^\dag$ ( the reverse process of $U$).
Besides,  the exact formula of de-cohering power in single qubit system makes the
comparison between  de-cohering power and generalized cohering power possible.
According to \eqref{eq:dl1} and \eqref{eq:dr}, the de-cohering power and generalized cohering power
of unitary operation on single qubit system are not equal in general, which is different from
the relationship between cohering power and generalized cohering power.

\begin{prop}
For any unitary operations $U$  on a single qubit system,   $\sD_{l_1}(U)$ and $\widehat\sD_{l_1}(U)$ are not equal in general, that is, there exist a
unitary operation $U_0$ such that $\sD_{l_1}(U_0)
<\widehat\sD_{l_1}(U_0)$.
\end{prop}

\begin{proof}
In single qubit system, $\sD_{l_1}(U)
=1-||a|^2-|b|^2|$ and $\widehat{\sD}_{l_1}(U)=\widehat{\sC}_{l_1}(U^\dag)=\sC_{l_1}(U^\dag)=2|ab|$ where
$\widehat{\sC}_{l_1}(U^\dag)=\sC_{l_1}(U^\dag)$ comes from the fact that cohering power coincides with generalized cohering power in
single qubit case \cite{Bu2015}.
Thus it is easy to take an unitary $U_0$
 such that
$\sD_{l_1}(U_0)
<\widehat\sD_{l_1}(U_0)$.

\end{proof}

\begin{prop}
For unitary operations $U$  on single qubit system,   $\sD_{r}(U)$ and $\widehat\sD_{r}(U)$  are not equal, that is, there exist a
unitary operation $U_0$ such that $\sD_{r}(U_0)
<\widehat\sD_{r}(U_0)$.
\end{prop}
\begin{proof}
Since the generalized de-cohering power need to take maximization over all quantum states, it is difficult to get
exact value of $\widehat\sD_{r}$. Thus, a lower bound of the generalized de-cohering power is expected instead of the
exact value.
Consider the following unitary operation and quantum state,
\begin{equation*}\label{counterU}
U_{0}=
\left( \begin{array}{ccc}
  0.5645 + 0.6351i & 0.4141 + 0.3264i\\
  -0.1452 + 0.5069i  & -0.0868 - 0.8452i
\end{array} \right),
\end{equation*}

\begin{equation*}\label{7}
\rho_{0}=
\left( \begin{array}{ccc}
    0.7063    &    0.4338 - 0.1360i\\
   0.4338 + 0.1360i &  0.2937
\end{array} \right),
\end{equation*}
then $\sD_r(U_0)\approx0.7053$ is strictly less than $[\cC_{r}(\rho_{0})-\cC_{r}(U_{1}\rho_{1} U^{\dag}_{1})]\approx
 0.8327$. As $\widehat{\sD}_r(U_0)\geq[\cC_r(\rho)-\cC_r(U_0\rho U^\dag_0)]$, then we prove the result.

\end{proof}

In view of the definition of $\widehat{\sD}$, $\widehat{\sD}(\Phi)=0$ implies
that $\cC(\rho)\leq \cC(\Phi(\rho))$ for any quantum state, that is, quantum operation
will not decrease coherence of any input state. Here, we investigate the set of quantum operations
with zero generalized de-cohering power,
\begin{eqnarray}
NDO_{l_1}&=&\set{\Phi:\widehat{\sD}_{l_1}(\Phi)=0},\\
NDO_r&=&\set{\Phi:\widehat{\sD}_r(\Phi)=0}.
\end{eqnarray}
 Note  that the set  $NDO_{l_1}$( resp. $NDO_r$) is the set of all quantum  operations that will not decrease the coherence of any state in
$l_1$ norm of coherence (resp. relative entropy of coherence).
It is easy to give some quantum operations that belongs to $NDO_{l_1}$ (or $NDO_{r}$), for example, 
take the quantum operation $\Phi$ with Kraus operators $\set{K_i}_i$,  where $K_i=\ket{\Psi}\!\bra{i}$ and $\ket{\Psi}$ is a maximally coherent state, then 
$\Phi$ maps any quantum state to maximally coherent $\ket{\Psi}$. It seems that there is no close relation between $NDO_{l_1}$ (or $NDO_{r}$) and
IO, MIO, as there exists coherence breaking operations \cite{BuE2016} map any quantum state to incoherent states.

\section{Comparison between cohering power and decohering power}\label{sec:comp}
It has been proved that
the cohering power of qubit unitary operations is equal to de-cohering power
in the skew information coherence   \cite{Mani2015}.
Here, we  consider the relationship between
cohering power and de-cohering power  for the unitary operations defined by $l_1$ norm and relative entropy respectively.

\begin{prop}
For any unitary operation $U$ on a single qubit system, the
cohering power is always larger than de-cohering power in $l_1$ norm, that
is $\sC_{l_1}(U)\geq \sD_{l_1}(U)$. However, this relationship does not hold for unitary operations
in higher-dimensional space.
\end{prop}

\begin{proof}
Since $U$ can be written as $ U= e^{i \varphi}\Br{\begin{array}{ccc}
a & b  \\
-b^{\ast} & a^{\ast}
\end{array}}
$ with
$|a|^{2}+|b|^{2}=1$, the cohering power of $U$ is $\sC_{l_1}(U)=2|ab|$.
And by the definition of the de-cohering power, we have
\begin{eqnarray}
\sD_{l_1}(U)
=1-||a|^2-|b|^2|\leq 2|ab|=\sC_{l_1}(U)
\end{eqnarray}

Take $U$ on d-dimensional system with $d\geq 3$ as following
\begin{eqnarray*}
U=\frac{\sqrt{2}}{2}(\proj{1}+\proj{2}+\ket{1}\!\bra{2}-\ket{2}\!\bra{1})+\sum^d_{k>2}\proj{k},
\end{eqnarray*}
then $\sC_{l_1}(U)=1$ and for maximally coherent state $\ket{\psi}=\frac{1}{\sqrt{d}}\sum_ke^{i\theta_k}\ket{k}$,
$U\ket{\psi}=\frac{1}{\sqrt{2d}}(e^{i\theta_1}+e^{i\theta_2})\ket{1}+\frac{1}{\sqrt{2d}}(e^{i\theta_1}-e^{i\theta_2})\ket{2}
+\frac{1}{\sqrt{d}}\sum^d_{k>2}e^{i\theta_k}\ket{k}$, which implies that 
\begin{eqnarray*}
\sD_{l_1}(U)&=&d-1-\min_{\ket{\psi}\in\cM}\cC_{l_1}(U\ket{\psi})\\
&=&(2-\sqrt{2})(2-\frac{2-\sqrt{2}}{d}).
\end{eqnarray*}
 Moreover,
$\sD_{l_1}(U)$ is larger than $(2-\sqrt{2})(2-\frac{2-\sqrt{2}}{3})$  when
$d\geq 3$. It is easy to check that $(2-\sqrt{2})(2-\frac{2-\sqrt{2}}{3})$
is strictly larger than 1. Thus, we have $\sC_{l_1}(U)<\sD_{l_1}(U)$.
\end{proof}

\begin{cor}
For any unitary operation $U$ on a single qubit system,  we have the following relationship
\begin{eqnarray}
\label{re1}\widehat{\sD}_{l_1}(U)=\widehat{\sC}_{l_1}(U)=\sC_{l_1}(U)\geq \sD_{l_1}(U)
\end{eqnarray}
\end{cor}
\begin{proof}
To prove \eqref{re1}, we only need to prove $\widehat{\sD}_{l_1}(U)=\widehat{\sC}_{l_1}(U)$.
Since $U$ can be written as $ U= e^{i \varphi}\Br{\begin{array}{ccc}
a & b  \\
-b^{\ast} & a^{\ast}
\end{array}}
$ with
$|a|^{2}+|b|^{2}=1$, the cohering power of $U$ is $\sC_{l_1}(U)=2|ab|=\sC_{l_1}(U^\dag)$.
As  we have proved that
$\widehat{\sD}_{l_1}(U)=\widehat{\sC}_{l_1}(U^\dag)$ in Proposition \ref{fom:de}
and $\widehat{\sC}_{l_1}(U)=\sC_{l_1}(U)$ \cite{Bu2015}, we have $\widehat{\sD}_{l_1}(U)=\widehat{\sC}_{l_1}(U^\dag)=\sC_{l_1}(U^\dag)=\sC_{l_1}(U)$.

\end{proof}

\begin{prop}
For any unitary operation $U$ on a single qubit system, the
cohering power is always larger than de-cohering power in relative entropy coherence, that
is $\sC_{r}(U)\geq \sD_{r}(U)$.
\end{prop}
\begin{proof}
Since $U$ can be written as $ U= e^{i \varphi}\Br{\begin{array}{ccc}
a & b  \\
-b^{\ast} & a^{\ast}
\end{array}}
$ with
$|a|^{2}+|b|^{2}=1$, the cohering power of
$\sC_r(U)=S(|a|^2,|b|^2)$. And the de-cohering power
of $U$ is $\sD_r(U)=1-S(\frac{1}{2}+|ab|,\frac{1}{2}+|ab|)$.
Thus, $\sC_{r}(U)\geq \sD_{r}(U)$ is equivalent to
$S(|a|^2,|b|^2)+S(\frac{1}{2}+|ab|,\frac{1}{2}-|ab|)\geq1$.
Due to Lemma \ref{appendix:2} in Appendix, we get the result.

\end{proof}

Although we have proved that  $\sC_{r}(U)\geq \sD_{r}(U)$ and
$\widehat{\sD}_{r}(U)=\widehat{\sC}_{r}(U^\dag)$, we cannot get the similar result like
\eqref{re1} as cohering power $\sC_{r}(U)$ and $\widehat{\sC}_{r}(U)$
are not equal even in single qubit case \cite{Bu2015}. Besides, as the explicit formula for 
de-cohering power $\sD_r$ in higher dimensions is still unknown even for unitary operations, the 
relationship between $\sD_r$ and $\sC_r$ remains  to be identified.

\section{Conclusion}\label{sec:con}
In this work, we have investigated the cohering power and de-cohering power which are defined to quantify the
ability of quantum operations to produce coherence and erase coherence respectively.
It has been proved that cohering power $\sC_{l_1}$ and generalized cohering power $\widehat{\sC}_{l_1}$ are equal for single qubit unitary operations \cite{Bu2015,Diaz2015}.
In this work, we prove that this statement is also true for any quantum operation on single qubit systems, which
implies the monotonicity of $l_1$ norm of coherence under  MIO
on single qubit systems. However, we show that $l_1$ norm of coherence is not monotone under DIO or MIO in higher dimensional space.
Thus we give a complete answer to the open problem about the monotonicity of $l_1$ norm of coherence under MIO proposed in \cite{Chitambar2016,ChitambarPRA2016}.
And the non-monotonicity of $l_1$ norm coherence implies that
$l_1$ norm is not contracting under CPTP maps. Contracting under CPTP maps is a basic property
for norms to be coherence measures \cite{Baumgratz2014},
thus it is amazing  that $l_1$ norm can be employed to quantify coherence
although it does not have
this property.
Besides, we investigate the connections between the sets of operations with zero generalized cohering power $NIO_{l_1}$ and $NIO_{r}$ with IO and MIO:
 $IO\subsetneq NIO_{l_1}=MIO=NIO_r$ in single qubit case and
$IO\subsetneq NIO_{l_1}\subsetneq MIO=NIO_r$ in higher dimensions;  MIO is just the set of all quantum  operations that will not increase the coherence of all states in relative entropy measure.
Moreover, we derive the exact formula of de-cohering power of single unitary operations.
By a comparison between de-cohering power and generalized de-cohering power, we have shown that they are not equal in general
which is different from the coincidence between cohering power
and generalized cohering power in single qubit systems.
Furthermore, we compare  cohering power and de-cohering power defined in $l_1$ norm  and
relative entropy,  and find that
cohering power is usually larger than de-cohering power for unitary operations on single qubit systems.

The results in this work present a new approach to study the free operations in the resource of coherence by cohering power  and therefore, are of great value to our understanding of IO, MIO and DIO proposed in \cite{Baumgratz2014,Chitambar2016,ChitambarPRA2016}.
However, more work is needed in this context. For example,  it will be useful to obtain the relationship between cohering power $\sC_{l_1}$ and $\sC_{r}$ (or
de-cohering power $\sD_{l_1}$ and $\sD_{r}$) for any quantum operation.
Another important question for future studies is to determine the relationship between
cohering power and de-cohering power for any quantum operations on higher dimensions.


\smallskip
\noindent
\begin{acknowledgments}
This work is supported by the Natural Science Foundations of China (Grants No.11171301 and No. 10771191) and the Doctoral Programs Foundation of the Ministry of Education of China (Grant No. J20130061).

\end{acknowledgments}

\appendix

\section{Several useful Lemmas}

\begin{lem}\label{appendix:2}
The function $H(x):=-x\log_2 x-(1-x)\log_2(1-x)$ with $x\in [0,1]$ satisfy
\begin{eqnarray}
H(x)+H(\frac{1}{2}+\sqrt{x(1-x)})\geq 1,
\end{eqnarray}
for any $x\in[0,1]$.
\end{lem}
\begin{proof}
 To prove this inequality is equal to
prove
\begin{eqnarray*}
-x\ln x-(1-x)\ln(1-x)-t\ln t-(1-t)\ln(1-t)\geq \ln2
\end{eqnarray*}
with $t=\frac{1}{2}+\sqrt{x(1-x)}$.
Since the symmetry of the formula, we only need to consider the
the case $x\in[0.5,1]$. As variables $x,t$ satisfy $(x-1/2)^2+(t-1/2)^2=1/2$, we
change the variables $x,t$ to $x=\frac{1+\cos\theta}{2}$ and
$t=\frac{1+\sin\theta}{2}$ with $\theta\in[0,\pi /2]$.
Then we prove the following the inequality:
\begin{eqnarray*}
 f(\theta)= &&-(\frac{1+\cos\theta}{2})\ln (\frac{1+\cos\theta}{2})
 -(\frac{1-\cos\theta}{2})\ln(\frac{1-\cos\theta}{2})\\
 &&-(\frac{1-\sin\theta}{2})\ln (\frac{1-\sin\theta}{2})
 -(\frac{1+\sin\theta}{2})\ln (\frac{1+\sin\theta}{2})\\
&& -\ln2\geq 0
\end{eqnarray*}
with $\theta\in[0,\pi/2]$.
Differentiate $f(\theta)$ with respect to $\theta$, then
\begin{eqnarray*}
\frac{df}{d\theta}
&=&\frac{1}{2}\Br{\sin\theta\ln\frac{1+\cos\theta}{1-\cos\theta}-
\cos\theta\ln\frac{1+\sin\theta}{1-\sin\theta}}\\
&=&\frac{1}{2\sin\theta\cos\theta}
\Br{\frac{1}{\cos\theta}\ln\frac{1+\cos\theta}{1-\cos\theta}
-\frac{1}{\sin\theta}\ln\frac{1+\sin\theta}{1-\sin\theta}}.
\end{eqnarray*}
Consider the function $g(s)=\frac{1}{s}\ln\frac{1+s}{1-s}$ with $s\in[0,1]$.
Then $\frac{dg}{ds}=\frac{1}{s^2}[\ln(1-s)+\frac{1}{1-s}-(\ln(1+s)+\frac{1}{1+s})]>0$, that is, $g(s)$ is a monotonous function.
Thus

(1) when $\theta\in[0,\pi/4]$, then $\cos\theta\geq \sin\theta$.
As the function $g(s)$ is monotonous, thus
$\frac{df}{d\theta}\geq0$.

(2) when $\theta\in[\pi/4,\pi/2]$, then $\cos\theta\leq \sin\theta$.
As the function $g(s)$ is monotonous, thus
$\frac{df}{d\theta}\leq0$.

Therefore, $\min_{\theta\in[0,\pi/2]}f(\theta)
=\min\set{f(0),f(\pi/2)}=0$.
\end{proof}


\bibliographystyle{apsrev4-1}
 \bibliography{coh-break-lit}

\end{document}